\documentclass[12pt]{amsart}
\usepackage{amssymb}
\usepackage{color}
\usepackage{amsmath,epic,curves,amscd}
\usepackage[english]{babel}
\usepackage{graphicx}
\usepackage{comment}
\usepackage{appendix}
\usepackage{mathdots}
\usepackage[all]{xy}
\newtheorem*{Lem*}{Lemma}

\theoremstyle{remark}

\renewenvironment{proof}{\noindent{\it Proof. \hskip0pt}}
                      {$\square$\par\medskip}

\begin{document}
\baselineskip 6.0 truemm
\parindent 1.5 true pc

\newcommand\xx{{\text{\sf X}}}
\newcommand\lan{\langle}
\newcommand\ran{\rangle}
\newcommand\tr{{\text{\rm Tr}}\,}
\newcommand\ot{\otimes}
\newcommand\ol{\overline}
\newcommand\join{\vee}
\newcommand\meet{\wedge}
\renewcommand\ker{{\text{\rm Ker}}\,}
\newcommand\image{{\text{\rm Im}}\,}
\newcommand\id{{\text{\rm id}}}
\newcommand\tp{{\text{\rm tp}}}
\newcommand\pr{\prime}
\newcommand\e{\epsilon}
\newcommand\la{\lambda}
\newcommand\inte{{\text{\rm int}}\,}
\newcommand\ttt{{\text{\rm t}}}
\newcommand\spa{{\text{\rm span}}\,}
\newcommand\conv{{\text{\rm conv}}\,}
\newcommand\rank{\ {\text{\rm rank of}}\ }
\newcommand\re{{\text{\rm Re}}\,}
\newcommand\ppt{\mathbb T}
\newcommand\rk{{\text{\rm rank}}\,}
\newcommand\SN{{\text{\rm SN}}\,}
\newcommand\SR{{\text{\rm SR}}\,}
\newcommand\HA{{\mathcal H}_A}
\newcommand\HB{{\mathcal H}_B}
\newcommand\HC{{\mathcal H}_C}
\newcommand\CI{{\mathcal I}}
\newcommand{\bra}[1]{\langle{#1}|}
\newcommand{\ket}[1]{|{#1}\rangle}
\newcommand\cl{\mathcal}
\newcommand\idd{{\text{\rm id}}}
\newcommand\OMAX{{\text{\rm OMAX}}}
\newcommand\OMIN{{\text{\rm OMIN}}}
\newcommand\diag{{\text{\rm Diag}}\,}
\newcommand\calI{{\mathcal I}}
\newcommand\bfi{{\bf i}}
\newcommand\bfj{{\bf j}}
\newcommand\bfk{{\bf k}}
\newcommand\bfl{{\bf l}}
\newcommand\bfp{{\bf p}}
\newcommand\bfq{{\bf q}}
\newcommand\bfzero{{\bf 0}}
\newcommand\bfone{{\bf 1}}
\newcommand\sa{{\rm sa}}
\newcommand\ph{{\rm ph}}
\newcommand\phase{{\rm ph}}
\newcommand\res{{\text{\rm res}}}
\newcommand{\algname}[1]{{\sc #1}}
\newcommand{\Setminus}{\setminus\hskip-0.2truecm\setminus\,}
\newcommand\calv{{\mathcal V}}
\newcommand\calg{{\mathcal G}}
\newcommand\calt{{\mathcal T}}
\newcommand\calvnR{{\mathcal V}_n^{\mathbb R}}
\newcommand\D{{\mathcal D}}
\newcommand\C{{\mathcal C}}

\title{Construction of three-qubit biseparable states distinguishing kinds of entanglement in a partial separability classification}

\author{Kyung Hoon Han and Seung-Hyeok Kye}
\address{Kyung Hoon Han, Department of Data Science, The University of Suwon, Gyeonggi-do 445-743, Korea}
\email{kyunghoon.han at gmail.com}
\address{Seung-Hyeok Kye, Department of Mathematics and Institute of Mathematics, Seoul National University, Seoul 151-742, Korea}
\email{kye at snu.ac.kr}
\thanks{Both KHH and SHK were partially supported by NRF-2017R1A2B4006655, Korea}



\begin{abstract}
We construct seven kinds of three-qubit biseparable states to show
that every class of biseparable states in the partial separability classification
proposed by Szalay and K\" ok\' enyesi [S. Szalay and Z. K\" ok\' enyesi, Phys. Rev. A {\bf
86}, 032341 (2012)] is nonempty.
\end{abstract}

\maketitle

\section{Introduction}

Entanglement is considered as one of the main resources in the current quantum information theory and it is an important research topic to detect entanglement. A state is said to be separable if it is a convex combination of pure product states and entangled if it is not separable. In multipartite systems, we have various kinds of separability and entanglement according to partitions of systems, and several authors \cite{{dur-cirac-tarrach},{dur-cirac},{coffman},{dur-vidal},{acin},{seevinck-uffink}}
classified them. A multipartite state is said to be biseparable if it is a convex combination of
states which are separable with respect to bipartitions of systems.
The class of all biseparable states makes the largest classes of separable states, and the classification problem is how to classify biseparable states.

Most recently, Szalay and K\" ok\' enyesi \cite{sz2012} proposed a finer classification than the previous ones for tripartite systems, but some of the classes in this classification are not known to be nonempty (see also \cite{sz2015,sz2018}).
In this paper we give analytic examples of three-qubit biseparable states to show that they are nonempty.
We construct requested examples
among {\sf X}-shaped states whose entries are zero except for diagonal and antidiagonal entries.
We had considered in \cite{han_kye_tri} entanglement witnesses corresponding to various notions of separability and characterized \cite{han_kye_tri,han_kye_optimal} such entanglement witnesses for {\sf X}-shaped self-adjoint multiqubit matrices as well as various kinds of separability. These results are the main tools to construct examples.

Recall that a tripartite state is said to be $A$-$BC$ separable if it is separable with respect to the
bipartition $A$:$BC$ of the systems $A$, $B$, and $C$. We also define $B$-$CA$ and $C$-$AB$ separability
similarly. We first construct a three-qubit state $\varrho_1$ which is
\begin{itemize}
\item[(1a)]
$A$-$BC$ separable,
\item[(1b)]
a mixture of $B$-$CA$ or $C$-$AB$ separable states,
\item[(1c)]
neither $B$-$CA$ separable nor $C$-$AB$ separable.
\end{itemize}
This shows that the class ${\mathcal C}^{2,6,1}$ is nonempty with the notation in \cite{sz2012}. The classes ${\mathcal C}^{2,6,2}$ and ${\mathcal C}^{2,6,3}$ are also nonempty by similar examples obtained from the flip operations on the systems $A$, $B$, and $C$.
We also construct a three-qubit state $\varrho_2$ which is
\begin{itemize}
\item[(2a)]
a mixture of $A$-$BC$ or $B$-$CA$ separable states,
\item[(2b)]
a mixture of $B$-$CA$ or $C$-$AB$ separable states,
\item[(2c)]
a mixture of $C$-$AB$ or $A$-$BC$ separable states,
\item[(2d)]
neither $A$-$BC$ separable nor $B$-$CA$ separable nor $C$-$AB$ separable.
\end{itemize}
This example of a state will show that the class ${\mathcal C}^{2,4}$ in \cite{sz2012} is nonempty. Finally, we construct an example of a three-qubit state $\varrho_3$ which is
\begin{itemize}
\item[(3a)]
a mixture of $A$-$BC$ or $B$-$CA$ separable states,
\item[(3b)]
a mixture of $C$-$AB$ or $A$-$BC$ separable states,
\item[(3c)]
not $A$-$BC$ separable,
\item[(3d)]
not a mixture of $B$-$CA$ or $C$-$AB$ separable states,
\end{itemize}
which gives us an example of states belonging to the class ${\mathcal C}^{2,3,1}$ in \cite{sz2012}. Similar examples in
the classes ${\mathcal C}^{2,3,2}$ and ${\mathcal C}^{2,3,3}$ can be exhibited.

\section{Construction}

Three-qubit states are considered as $8\times 8$ matrices, by
the identification $M_8=M_2\otimes M_2\otimes M_2$ with the lexicographic order
of indices in the tensor product. Therefore,
a three-qubit {\sf X}-shaped Hermitian matrix is of the form
$$
\xx(a,b,z)= \left(
\begin{matrix}
a_1 &&&&&&& z_1\\
& a_2 &&&&& z_2 & \\
&& a_3 &&& z_3 &&\\
&&& a_4&z_4 &&&\\
&&& \bar z_4& b_4&&&\\
&& \bar z_3 &&& b_3 &&\\
& \bar z_2 &&&&& b_2 &\\
\bar z_1 &&&&&&& b_1
\end{matrix}
\right),
$$
for $a,b\in\mathbb R^4$ and $z\in\mathbb C^4$.
An {\sf X}-shaped state is also called an \xx-state for brevity.

By the result in \cite{han_kye_optimal} (see Proposition 5.2 therein), we see that
a three-qubit {\sf X}-state $\varrho = \xx(a,b,z)$ is $A$-$BC$
separable if and only if it is of positive partial transpose as a
bipartite state with respect to the partition $A$:$BC$.
Therefore, $\varrho$ is $A$-$BC$ separable if and only if the conditions
\begin{equation}\label{1|23}
\begin{aligned}
&\min\{ \sqrt{a_1b_1}, \sqrt{a_4b_4} \} \ge \max\{ |z_1|,|z_4| \}, \\
&\min\{ \sqrt{a_2b_2}, \sqrt{a_3b_3} \} \ge \max\{ |z_2|,|z_3| \}
\end{aligned}
\end{equation}
are satisfied. Similarly, we also see that $\varrho$ is $B$-$CA$ separable if and only if
\begin{equation}\label{2|31}
\begin{aligned}
&\min\{ \sqrt{a_1b_1}, \sqrt{a_3b_3} \} \ge \max\{ |z_1|,|z_3| \}, \\
&\min\{ \sqrt{a_2b_2}, \sqrt{a_4b_4} \} \ge \max\{ |z_2|,|z_4| \}\
\end{aligned}
\end{equation}
hold and $C$-$AB$ separable if and only if
\begin{equation}\label{3|12}
\begin{aligned}
&\min\{ \sqrt{a_1b_1}, \sqrt{a_2b_2} \} \ge \max\{ |z_1|,|z_2| \}, \\
&\min\{ \sqrt{a_3b_3}, \sqrt{a_4b_4} \} \ge \max\{ |z_3|,|z_4| \}.
\end{aligned}
\end{equation}

We begin with the two {\sf X}-states
$$
\begin{aligned}
\varrho_{1B} &= \xx((0,1,0,1),(0,1,0,1),(0,1,0,0)), \\
\varrho_{1C} &= \xx((0,0,1,1),(0,0,1,1),(0,0,1,0)),
\end{aligned}
$$
which are $B$-$CA$ and $C$-$AB$ separable by (\ref{2|31}) and (\ref{3|12}), respectively.
The mixture of them
$$
\varrho_1=\varrho_{1B}+\varrho_{1C}=\xx((0,1,1,2),(0,1,1,2),(0,1,1,0))
$$
meets the condition (\ref{1|23}), but violates both (\ref{2|31}) and (\ref{3|12}).
Therefore, the three-qubit {\sf X}-state $\varrho_1$ satisfies all the conditions (1a), (1b) and (1c).

For the second example, we consider the {\sf X}-states
$$
\begin{aligned}
\varrho_{2A0} &= \xx((0,1,2,0),(0,1,2,0),(0,1,1,0)),\\
\varrho_{2A1} &= \xx((0,2,1,0),(0,2,1,0),(0,1,1,0)),
\end{aligned}
$$
which are $A$-$BC$ separable. We also consider the states
$$
\begin{aligned}
\varrho_{2B0} &= \xx((0,1,0,2),(0,1,0,2),(0,0,0,1)),\\
\varrho_{2B1} &= \xx((0,2,0,1),(0,2,0,1),(0,1,0,0)),\\
\varrho_{2C0} &= \xx((0,0,1,2),(0,0,1,2),(0,0,0,1)),\\
\varrho_{2C1} &= \xx((0,0,2,1),(0,0,2,1),(0,0,1,1)).
\end{aligned}
$$
We note that both $\varrho_{2B0}$ and $\varrho_{2B1}$ are $B$-$CA$ separable;
both $\varrho_{2C0}$ and $\varrho_{2C1}$ are $C$-$AB$ separable; however, the state
$$
\begin{aligned}
\varrho_2
&=\xx((0,2,2,2),(0,2,2,2),(0,1,1,1))\\
&=\varrho_{2A0} + \varrho_{2B0}\\
& = \varrho_{2B1} + \varrho_{2C1}\\
&= \varrho_{2C0} + \varrho_{2A1}
\end{aligned}
$$
violates all the conditions (\ref{1|23}), (\ref{2|31}), and (\ref{3|12}) and so we conclude that the state
$\varrho_2$ satisfies all the conditions (2a), (2b), (2c), and (2d).

In order to construct the third example $\varrho_3$ of biseparable state, we need the characterization
(see \cite{han_kye_tri}, Theorem 6.2) of entanglement witnesses
among {\sf X}-shaped three-qubit Hermitian matrices $W = \xx(s,t,u)$: Recall that
$\lan W, \varrho \ran \ge 0$ for all $A$-$BC$ separable states $\varrho$ if and only if the relations
$$
\begin{aligned}
\sqrt{s_1t_1}+\sqrt{s_4t_4} &\ge |u_1|+|u_4|,\\
\sqrt{s_2t_2}+\sqrt{s_3t_3} &\ge |u_2|+|u_3|
\end{aligned}
$$
hold, where $\lan A,B\ran=\tr(A^\ttt B)$ is the usual bilinear pairing between matrices.
Similarly, we have $\lan W, \varrho \ran \ge 0$ for all $B$-$CA$ separable states $\varrho$ if and only if
\begin{equation}\label{W2|31}
\begin{aligned}
\sqrt{s_1t_1} + \sqrt{s_3t_3} &\ge |u_1|+|u_3|,\\
\sqrt{s_2t_2}+\sqrt{s_4t_4} &\ge |u_2|+|u_4|
\end{aligned}
\end{equation}
are satisfied. Finally, $\lan W, \varrho \ran \ge 0$ for all $C$-$AB$ separable states $\varrho$ if and only if
\begin{equation}\label{W3|12}
\begin{aligned}
\sqrt{s_1t_1}+\sqrt{s_2t_2} &\ge |u_1|+|u_2|,\\
\sqrt{s_3t_3}+\sqrt{s_4t_4} &\ge |u_3|+|u_4|.
\end{aligned}
\end{equation}

\begin{Lem*}\label{lemma}
Let $\varrho$ be a three-qubit state with the {\sf X}-part $\xx(a,b,z)$.
Then we have the following:
\begin{enumerate}
\item[(i)] if $\varrho$ is a mixture of $B$-$CA$ or $C$-$AB$ separable states, then
\begin{equation}\label{2|31,3|12}
\min\{ \sqrt{a_1b_1}+\sqrt{a_4b_4}, \sqrt{a_2b_2}+\sqrt{a_3b_3} \} \ge \max\{ |z_1|+|z_4|, |z_2|+|z_3| \};
\end{equation}
\item[(ii)] if $\varrho$ is a mixture of $C$-$AB$ or $A$-$BC$ separable states, then
$$
\min\{ \sqrt{a_1b_1}+\sqrt{a_3b_3}, \sqrt{a_2b_2}+\sqrt{a_4b_4} \} \ge \max\{ |z_1|+|z_3|, |z_2|+|z_4| \};
$$
\item[(iii)] if $\varrho$ is a mixture of $A$-$BC$ or $B$-$CA$ separable states, then
$$
\min\{ \sqrt{a_1b_1}+\sqrt{a_2b_2}, \sqrt{a_3b_3}+\sqrt{a_4b_4} \} \ge \max\{ |z_1|+|z_2|, |z_3|+|z_4| \}.
$$
\end{enumerate}
\end{Lem*}

\begin{proof}
Suppose that $a_i, b_i>0$ for all $i=1,2,3,4$.
We consider \xx-shaped three-qubit Hermitian matrices $W_1$ and $W_2$ as
$$
\begin{aligned}
W_1=&\xx((0,\sqrt{b_2 \over a_2},\sqrt{b_3 \over a_3},0),(0,\sqrt{a_2 \over b_2},\sqrt{a_3 \over b_3},0),(-e^{-{\rm i}\theta_1},0,0,-e^{-{\rm i}\theta_2})),\\
W_2=&\xx((\sqrt{b_1 \over a_1},0,0,\sqrt{b_4 \over a_4}),(\sqrt{a_1 \over b_1},0,0,\sqrt{a_4 \over b_4}),(0,-e^{-{\rm i}\theta_2},-e^{-{\rm i}\theta_3},0))
\end{aligned}
$$
for $\theta_i = \arg z_i$.
Then $W_1$ and $W_2$ obey both (\ref{W2|31}) and (\ref{W3|12}).
Therefore, we have
$$
0 \le {1 \over 2} \lan \varrho, W_1 \ran = \sqrt{a_2b_2}+\sqrt{a_3b_3}-|z_1|-|z_4|
$$
and
$$
0 \le {1 \over 2} \lan \varrho, W_2 \ran = \sqrt{a_1b_1}+\sqrt{a_4b_4}-|z_2|-|z_3|.
$$
For general $a_i, b_i \ge 0$, we apply the above argument to $\varrho+\varepsilon I$ for $\varepsilon>0$ and take $\varepsilon \to 0$,
to complete the proof of statement (i).

By the flip on the first and second subsystems of $M_2 \otimes M_2 \otimes M_2$, we see that (i) implies (ii).
Similarly, we consider the flip on the first and third subsystems of $M_2 \otimes M_2 \otimes M_2$, to see that
(i) implies (iii).
\end{proof}

Now we consider the states
$$
\begin{aligned}
\varrho_{3A0} &= \xx((0,1,3,0),(0,1,3,0),(0,1,1,0)),\\
\varrho_{3A1} &= \xx((0,2,2,0),(0,2,2,0),(0,2,0,0))
\end{aligned}
$$
which are $A$-$BC$ separable. We also consider the states
$$
\begin{aligned}
\varrho_{3B} &=\xx((0,1,0,1),(0,1,0,1),(0,1,0,1)),\\
\varrho_{3C} &=\xx((0,0,1,1),(0,0,1,1),(0,0,1,1)).
\end{aligned}
$$
Note that $\varrho_{3B}$ and $\varrho_{3C}$ are $B$-$CA$ and $C$-$AB$ separable, respectively.
Therefore, the state
$$
\begin{aligned}
\varrho_3
&=\xx((0,2,3,1),(0,2,3,1),(0,2,1,1))\\
&=\varrho_{3A0}+\varrho_{3B}\\
&=\varrho_{3A1}+\varrho_{3C}
\end{aligned}
$$
satisfies the conditions (3a) and (3b). The state $\varrho_3$ meets
the conditions (3c) and (3d), because it violates both the
relations (\ref{1|23}) and (\ref{2|31,3|12}). This completes the
construction of $\varrho_1$, $\varrho_2$, and $\varrho_3$.

\section{Conclusion}

We have exhibited analytic examples of three-qubit biseparable
states which belong to the seven classes ${\mathcal C}^{2,6,1}$,
${\mathcal C}^{2,6,2}$, ${\mathcal C}^{2,6,3}$, ${\mathcal
C}^{2,4}$, ${\mathcal C}^{2,3,1}$, ${\mathcal C}^{2,3,2}$, and
${\mathcal C}^{2,3,3}$, respectively. This shows that all the classes
in the partial separability classification of \cite{sz2012} are nonempty in
the tripartite case.

\section*{Acknowledgments}
This research was supported by Basic Science Research Program through the National Research Foundation of Korea funded by the Ministry of Education, Science and Technology (Grant No. NRF-2017R1A2B4006655).



\begin{thebibliography}{11}


\bibitem{dur-cirac-tarrach}
W. D\" ur, J. I. Cirac and R. Tarrach,
\it Separability and Distillability of Multiparticle Quantum Systems,
\rm Phys. Rev. Lett. {\bf 83} (1999), 3562--3565.

\bibitem{dur-cirac}
W. D\" ur and J. I. Cirac,
\it Classification of multi-qubit mixed states: separability and distillability properties,
\rm Phys. Rev. A {\bf 61} (2000), 042314.

\bibitem{coffman}
V. Coffman, J. Kundu and W. K. Wootters,
\it Distributed entanglement,
\rm Phys. Rev. A {\bf 61} (2000), 052306.

\bibitem{dur-vidal}
W. D\" ur, G. Vidal and J. I. Cirac,
\it Three qubits can be entangled in two inequivalent ways,
\rm Phys. Rev. A {\bf 62} (2000), 062314.

\bibitem{acin}
A. Acin, D. Bru\ss, M. Lewenstein and A. Sanpera,
\it Classification of Mixed Three-Qubit States,
\rm Phys. Rev. Lett. {\bf 87} (2001), 040401.

\bibitem{seevinck-uffink}
M. Seevinck and J. Uffink,
\it Partial separability and etanglement criteria for multiqubit quantum states,
\rm Phys. Rev. A {\bf 78} (2008), 032101.

\bibitem{sz2012}
S. Szalay and Z. K\" ok\' enyesi,
\it Partial separability revisited: Necessary and sufficient criteria,
\rm Phys. Rev. A {\bf 86} (2012), 032341.

\bibitem{sz2015}
S. Szalay,
\it Multipartite entanglement measures,
\rm Phys. Rev. A {\bf 92} (2015), 042329.

\bibitem{sz2018}
S. Szalay,
\it The classification of multipartite quantum correlation,
\rm J. Phys. A: Math. Theor. {\bf 51} (2018), 485302.

\bibitem{han_kye_tri}
K. H. Han and S.-H, Kye,
\it Various notions of positivity for bi-linear maps and applications to tri-partite entanglement,
\rm J. Math. Phys. {\bf 57} (2016), 015205.

\bibitem{han_kye_optimal}
K. H. Han and S.-H, Kye,
\it Construction of multi-qubit optimal genuine entanglement witnesses,
\rm J. Phys. A: Math. Theor. {\bf 49} (2016), 175303.
\end{thebibliography}
\end{document}